\documentclass[ecta,nameyear,draft,nolinenumbers]{econsocart}
\usepackage[thm_section]{macros}
\usepackage{microtype}
\let\oldFootnote\footnote
\newcommand\nextToken\relax

\renewcommand\footnote[1]{%
    \oldFootnote{#1}\futurelet\nextToken\isFootnote}

\newcommand\isFootnote{%
    \ifx\footnote\nextToken\textsuperscript{,}\fi}

\newcommand{\bY}{\mathbf{Y}}
\newcommand{\by}{\mathbf{y}}
\newcommand{\reg}{\mathrm{Regret}}
\newcommand{\TE}{\tau}
\newcommand{\hTE}{\hat{\tau}}
\newcommand{\risk}{\mathrm{Risk}}
\newcommand{\bx}{\mathbf{x}}
\newcommand{\bX}{\mathbf{X}}

\usepackage{clipboard}
\newclipboard{draft}

\begin{document}
\begin{frontmatter}

\title{Synthetic Control As Online Linear Regression}
\runtitle{Synthetic Control As Online Linear Regression}

\begin{aug}
\author[id=au1,addressref={add11,add1}]{\fnms{Jiafeng}~\snm{Chen}\ead[label=e1]
{jiafengchen@g.harvard.edu}}
\address[id=add1]{%
\orgdiv{Department of Economics},
\orgname{Harvard University}}
\address[id=add11]{%
\orgname{Harvard Business School}}

\end{aug}

\support{This paper would not have been possible without the generous encouragement and
guidance of Isaiah Andrews. I also thank the editor, two anonymous referees, Alberto
Abadie, Susan Athey, Xiaohong Chen, Avi Feller, Bruno Ferman, Wayne Gao, Edward Glaeser,
Sukjin Han, Patrick Kline, Scott Kominers, Apoorva Lal, Namrata Narain, Ross Mattheis,
David Ritzwoller, Brad Ross, Jonathan Roth, Bas Sanders, Suproteem Sarkar, Jesse Shapiro,
Neil Shephard, Jann Spiess, Elie Tamer, Jaume Vives-i-Bastida, Davide Viviano, Chris
Walker, and Jennifer Walsh for helpful discussions. I am also grateful for comments from
participants of the Harvard Graduate Student Workshop in Econometrics, the 2022 Synthetic
Control Workshop at the Princeton University Center for Statistics and  Machine Learning,
and the 2022 NBER Labor Studies Summer Institute Meeting. }

\begin{abstract} 
This paper notes a simple connection between synthetic control and online
learning.
 Specifically, we recognize synthetic control as an instance of
 \emph{Follow-The-Leader} (FTL). Standard results in
  online convex optimization then imply that, even when outcomes are chosen by an
  adversary, synthetic control predictions of counterfactual outcomes for the treated
  unit perform almost as well as an oracle weighted average of control units' outcomes.
  Synthetic control on differenced data performs almost as well as oracle weighted
  difference-in-differences, potentially making it an attractive choice in
  practice. We argue that this observation further supports the use of
  synthetic control estimators in comparative case studies.
\end{abstract}


\end{frontmatter}

\newpage

\section{Introduction}

Synthetic control \citep{abadie2003economic,abadie2015comparative} is an increasingly
popular method for causal inference among policymakers, private institutions, and social
scientists alike. In parallel, there is a rapidly growing methodological literature
providing statistical guarantees for synthetic control methods.\footnote{See the review by
\citet{abadie2021using} as well as the special section on synthetic control methods in the
\emph{Journal of the American Statistical Association} \citep{abadie2021introduction}.}
\Copy{introsent}{Existing results for synthetic control---and for
modifications thereof---are typically derived under a low-rank linear factor model or a
 vector autoregressive model of the outcomes \citep[see, among others,][]
 {abadie2010synthetic,ben2019synthetic,ben2021augmented,ferman2021synthetic,viviano2019synthetic}.}\footnote{Notably,
 like this paper, \citet{bottmer2021design} consider a design-based framework which
 conditions on the outcomes and considers randomness arising solely from assignment of the
 treated unit or the treatment time period. } While these statistical guarantees formally
 hold under these outcome models, a number of authors have expressed optimism that the
 synthetic control method is robust to these modeling assumptions.\footnote{For instance,
 \cite{ben2019synthetic} write, ``Outcome modeling can also be sensitive to model
 mis-specification, such as selecting an incorrect number of factors in a factor model.
 Finally, [... synthetic control] can be appropriate under multiple data generating
 processes (e.g., both the autoregressive model and the linear factor model) so that it is
 not necessary for the applied researcher to take a strong stand on which is correct.''
 \cite{jaume} write, ``Synthetic controls are intuitive, transparent, and produce reliable
 estimates for a variety of data generating processes.''}

On the other hand, in  empirical settings where synthetic control is commonly
applied---where the treated unit is an aggregate entity like a country or a U.S.
state---plausible outcome modeling may be challenging. \cite{manski2018right}, in
studying the effect of gun laws in the United States using state-level crime
rates, provocatively ask, ``what random process should be assumed to have generated the
existing United States, with its realized state-year crime rates?'' 
\Copy{tension}{Granted, the low-rank linear
factor model is a general class of data-generating processes and may even arise
under finer-grained models on the individual outcomes contained in the aggregate data 
\citep{shi2022assumptions}. But to pessimists and
skeptics, perhaps even such a model is implausible for the settings considered by many
synthetic control studies. Indeed, if practitioners were willing to
fully commit to an outcome model, perhaps they should estimate the outcome model
directly---e.g., use factor model-based methods
\citep{bai2002determining,bai2003inferential,xu2017generalized,athey2021matrix}---instead
of using synthetic control?}

As a result, existing methodological results seem to leave practitioners in a somewhat
awkward position. On the one hand, synthetic control is intuitively appealing, and it is
conjectured to have good properties under a variety of outcome models. On the other hand,
perhaps existing outcome models that have so far proved sufficiently analytically
tractable are not always compelling in common empirical settings. To address this tension,
this paper provides a few theoretical results and offers a novel interpretation of
synthetic control methods. In particular, we seek guarantees for synthetic control that do
not rely on any outcome model. Consequently, our results complement existing, model-based
ones.

It is unlikely that nontrivial guarantees  on the {performance} of synthetic control exist
without any structure on the outcomes. However, we {can} derive guarantees of synthetic
control's performance \emph{relative} to a class of alternatives, such as weighted
matching or weighted difference-in-differences (DID) estimators, which practitioners may
otherwise choose. Our first main result shows that, on average over \emph{hypothetical
treatment timings}, synthetic control predictions are never much worse than the
predictions made by any weighted matching estimator. Our second main result shows that the
same is true for synthetic control on differenced data versus any weighted DID estimator.
These results imply that if there is a weighted matching or DID estimator that performs
well, synthetic control likewise performs well.  To be clear, these \emph{regret
guarantees} average over hypothetical treatment timings, which can be interpreted as
expected loss under random treatment timing, a design-based assumption.

\Copy{intropractice}{Taken together, our results provide reassurances for practitioners,
as they offer justifications for synthetic control that do not rely on particular
statistical models of the outcomes. At least on average over hypothetical treatment
timings, regardless of outcomes, variations of synthetic control are competitive against
common estimators, such as weighted matching and weighted DID estimators.
Additionally, our second result introduces a novel version of synthetic control that is
competitive against DID. Since DID is extremely popular in practice
\citep{currie2020technology} and is thus a natural benchmark, this version of synthetic
control may be particularly attractive.}

\Copy{typo1}{We derive our results by casting prediction with panel data as an instance of
\emph{online convex optimization}, and by recognizing synthetic control as an online
  regression algorithm known as \emph{Follow-The-Leader} \citep[FTL, a name coined by][]
  {kalai2005efficient}.\footnote{For an introduction to online convex optimization, see
  \cite{hazan2019introduction}, \cite{orabona2019modern}, \cite{cesa2006prediction}, and
  \cite{shalev2011online}.}} Regret guarantees on FTL in the online convex optimization
  literature translate directly to guarantees for synthetic control against a class of
  alternative estimators. Since most results in online convex optimization have been
  derived under an adversarial model---where an imagined adversary generates the
  data---these results translate to guarantees on synthetic control without any structure
  on the outcome process.

This paper is perhaps closest to \cite{viviano2019synthetic}. They propose an ensemble
scheme to aggregate predictions from multiple predictive models, which can include
synthetic control, interactive fixed effects models, and random forests. Using results
from the online learning literature, \cite{viviano2019synthetic}'s ensemble scheme has the
no-regret property, making the ensemble predictions competitive against the predictions of
any fixed predictive model in the ensemble. Under sampling processes that yield good
performance for some predictive model in the ensemble, \cite{viviano2019synthetic} then
derive performance guarantees for the ensemble learner. In contrast, we study synthetic
control directly in the worst-case setting, and connect corresponding worst-case results
to guarantees on statistical risk in a design-based framework. We show
that synthetic control algorithms \emph{themselves} are no-regret online algorithms and
are in fact competitive against a wide class of matching or DID
estimators.

\Cref{sec:theory} sets up the notation and the decision
protocol and presents our main results for synthetic control. \Cref
{sec:extension} presents several extensions that show alternative guarantees on
modifications of synthetic control; in particular, we show that synthetic control on differenced data is competitive against a class of  difference-in-differences estimators. \Cref{sec:conc} concludes the paper.

\section{Setup and main results}
\label{sec:theory}

Consider a simple setup for synthetic control, following \cite {doudchenko2016balancing}.
There are $T$ time periods and $N + 1$ units. \Copy{tgen}{To simplify convergence rate
expressions, we assume $T > N$ unless noted otherwise, but this assumption is not strictly
necessary for our results.} Let unit $0$ be the only treated unit, first treated at some
time $S
\in \{1,\ldots, T\} \equiv[T]$. The other $N$ units are referred to as control units.
Since we observe the treated potential outcomes for the
treated unit after $S$, estimating causal effects for unit 0 amounts to predicting the
unobserved, post-$S$ untreated potential outcomes of this unit. Thus, we focus on untreated
potential outcomes.

 Let the full panel of untreated potential outcomes be $\bY$ with
 representative entry $y_{it}$, where (i) $\bY_ {1:s} = (y_{0t},\ldots,y_{Nt})_{t=1}^s$
 collects all untreated potential outcomes until and including time $s$, and (ii) $\by_t =
 (y_{1t},
\ldots, y_{Nt})'$ is the vector of control unit outcomes at time $t$. Additionally, we let
$\by(1) =
 (y_{1}(1),\ldots, y_{T}(1))'$ denote the treated potential outcomes of unit $0$, which
 are only observable for times $t \ge S$. Similarly, we let $\by(0) = (y_{01},\ldots, y_
 {0T})'$
 denote the untreated potential outcomes of unit $0$, which are observable for $t < S$.
 The analyst is tasked with predicting $y_ {0S}$ from observed data, which typically
 consist of pre-treatment outcomes of unit $0$ and outcomes of untreated units.
 \Copy{firstpara}{Like the main analysis in \cite {doudchenko2016balancing}, we do not
 consider covariates extensively, though \cref{sub:reg} considers matching on covariates
 as a form of regularization.\footnote{To extend our analysis to cases with covariates, at
 a minimum, we can interpret $\bY$ as the residuals of the untreated potential outcomes
 against some fixed regression function of the covariates, i.e. $y_ {it} = y_{it}^* -
 h_t(x_i)$, for fixed $h_t$ (perhaps estimated from auxiliary data), outcomes $y_{it}^*$,
 and covariate vectors $x_i$. The residualization is similar to Section 5.5 in
 \cite{doudchenko2016balancing} and expression (16) in \cite{abadie2021using}, but is
 stronger due to $h_t$ being fixed for different adversarial choices of $\bY$. Our results
 apply so long as these residuals obey the boundedness assumption $\norm{\bY}_\infty \le
 1$ that we impose later.} }

Synthetic control \citep{abadie2003economic,abadie2010synthetic}, in its basic form,
chooses some convex weights $\smash{\hat\theta_S}$ that minimize past prediction errors
\begin{equation}
	\hat\theta_S \in \argmin_{\theta \in \Theta} \sum_{t=1}^{S-1} (y_{0t} - \theta'\by_t)^2,
\label{eq:ftl}
\end{equation} where $\Theta \equiv \{(\theta_1,\ldots,\theta_N) \in \R^N \colon \theta_i
\ge 0, 1'\theta = 1\}$ is the simplex. For a one-step-ahead forecast for $y_{0S}$,
synthetic control outputs the weighted average $\hat y_{S} \equiv
 \hat\theta_S'\by_S$, and
 forms the treatment effect estimate $\hTE_{S} \equiv y_S(1) - \hat y_S$. 

 Theoretical guarantees for treatment effect estimates $\hat\tau_S$ often
 rely on statistical models of the outcomes $\bY$.\label{page:repeatsamp}
 \Copy{repeatsamp}{While synthetic control has good
  performance under a range of outcome models, one may still doubt whether these models are
  plausible---and whether the underlying repeated sampling thought
  experiments are
  appropriate---in the spirit of comments by \cite{manski2018right}.} In contrast to the
 usual outcome modeling approach, we instead consider a worst-case setting where the
 outcomes are generated by an adversary.\footnote{The adversarial framework, popular in
 online learning, dates to the works of \cite{hannan20164} and
 \cite{blackwell1956analog}.} Doing so has the appeal of giving decision-theoretic
 justification for methods while being entirely agnostic towards the data-generating
 process. Since a dizzying range of reasonable data-generating models and identifying
 assumptions are possible in panel data settings---yet perhaps none are unquestionably
 realistic---this worst-case view is valuable, and worst-case guarantees can be
 comforting.

\label{page:online}

\Copy{online}{
 In particular, we assume an adversary picks the outcomes $\bY$---or,
 equivalently, we derive results that hold uniformly over $\br{\bY: 
 \norm{\bY}_\infty
 \le 1}$. Specifically, we consider the following protocol
 between an analyst and an adversary:
 \begin{enumerate}[label=(P\arabic*),wide]
    \item \label{item:p1} The analyst commits to a class of linear prediction rules $\hat y_{t} \equiv
    f
    (\by_t;
    \theta_t(\bY_{1:t-1})) = \theta_t' \by_t$, parametrized by some $\theta_t \in \Theta$
    that may be chosen as a function of the past data $\bY_{1:t-1}$.
     We refer to the maps $\sigma
     \equiv \{\theta_t
     (\cdot) \colon t\in [T]\}$ as the analyst's \emph{strategy}. This means that if the
     treatment time $S$ is equal to $t$, then the analyst reports $\hat y_t$ as
     their prediction  for the untreated potential outcome at the first period after
     treatment.

    \item \label{item:p2} The adversary chooses the matrix of outcomes $\bY$. In order
    to obtain nontrivial bounds, we assume that the
     adversary cannot choose arbitrarily large outcomes, and without further loss of
     generality, we assume $\norm{\bY}_\infty \le 1$. Since we are interested in the worst
     case, the adversary may choose $\bY$ with knowledge of $\sigma$.

    \item \label{item:p3} The analyst suffers loss equal to squared prediction error at
    time $S$: i.e., $\ell(\hat y_{S}, y_{0S}) \equiv
    (\hat y_
    {S}- y_{0S})^2$.
\end{enumerate}
Under such a protocol, the analyst's average squared loss, averaging over
\emph{hypothetical values of $S$}, is%
\begin{equation}
  \frac{1}{T}\sum_{S=1}^T (y_
{0S} - \hat y_S)^2 = \frac{1}{T}\sum_{S=1}^T ( y_
{0S}- \theta_S'\by_S)^2 =  \E_{S \sim \Unif[T]}\bk{(y_{0S} - \hat y_S)^2}.
\label{eq:expectedlosspanel}
\end{equation}
Most results in this paper are guarantees in terms of the decision criterion 
\eqref{eq:expectedlosspanel} for synthetic control, where synthetic control \eqref{eq:ftl} is
viewed as a particular strategy
$\sigma$ under \cref{item:p1,item:p2,item:p3}.

As the second equality in \eqref{eq:expectedlosspanel} indicates, under an additional
assumption that treatment timing is \emph{uniformly random}, $S \sim \Unif [T]$, the
average loss over hypothetical treatment timings is equal to the expected squared loss
over $S$. This additional assumption is a design-based perspective
\citep{doudchenko2016balancing,bottmer2021design} on the panel causal inference problem.
This perspective enables us to interpret average prediction loss over hypothetical
treatment timings as expected prediction loss under the random treatment time $S$. The
latter can in turn be thought of as design-based risk. Uniformly random assignment of $S$
is restrictive, but we shall relax this requirement in
\cref{sub:nonuniform,sec:adaptivereg}.\footnote{The protocol \cref{item:p1,item:p2,item:p3} easily generalizes
when we replace $f(\by_t, \theta_t)$ with any known scalar function and $\ell
(\cdot,\cdot)$ with any loss function, so long as $\theta
\mapsto \ell(f (\by_t, \theta), y_{0t})$ is convex and bounded. Our results in 
\cref{sub:reg} allow for general loss functions.}

We now make clear the connection with online convex optimization \citep[see
Section 1.1 in][]{hazan2019introduction}. Online convex optimization works with the
following general protocol. Time $t$ increments sequentially for $T$ periods, and at time
$t$:
\begin{enumerate}[label=(O\arabic*),wide]
   \item \label{item:oco1} An
online player chooses some $\theta_t \in
\Theta$, where $\Theta \subset \R^d$ is a bounded convex set. The choice
$\theta_t$ may depend on the loss functions $\br{\ell_s: s < t}$ chosen by the adversary
in
the past. 
\item \label{item:oco2} After
$\theta_t$ is chosen, an adversary chooses a  loss function $\ell_t : \Theta \to
\R$ from some given set of loss functions, which may be further parametrized. These loss
functions are constrained to be convex and bounded but can otherwise be quite general.
They are often further constrained in order to obtain specific regret
results.
\item \label{item:oco3} The player
suffers loss $\ell_t(\theta_t)$ and observes $\ell_t(\cdot)$.\footnote{A closely related
setting where the player only observes $\ell_t(\theta_t)$ instead of the entire loss
function $\ell_t(\cdot)$ is known as \emph{bandit convex optimization} \citep[see Chapter
6 in][]{hazan2019introduction}, of which the adversarial multi-armed bandit problem 
\citep{robbins1952some,bubeck2012regret}
is a special case. }
The
player may update
their decision $\theta_{t+1}$ based on $\ell_1(\cdot),\ldots,\ell_t(\cdot)$.
\end{enumerate}
At the end of the game, the online player suffers total loss $\sum_{t=1}^T \ell_t
(\theta_t)$.
 
 Our setup of the panel prediction protocol, \cref{item:p1,item:p2,item:p3}, is then an
 instance of online convex optimization, \cref{item:oco1,item:oco2,item:oco3}. To see
 this, the most important step is to recognize that the analyst's loss 
 \eqref{eq:expectedlosspanel} is analogous to the online player's loss, and
 therefore to
 think of the analyst as making 
 {sequential decisions}
 where $\bY$ is sequentially revealed to them. This change in perspective relies on
 (i) our choice of decision criterion
 \eqref{eq:expectedlosspanel} and (ii) the fact that the analyst's decisions $\theta_t
 (\cdot)$ only require outcomes prior to $t$. Indeed, by fixing $\bY$ and considering the
 hypothetical values of $S=1,\ldots, T$ sequentially, we can treat the analyst as if they
 were solving an online problem and learning from data in the past---even though, for any
 particular value of $S$, they are only confronted with a static, offline problem. To be
 clear, we are not considering some online version of synthetic control; the
 connection to online convex optimization comes from considering hypothetical, unrealized
 values of $S$.

 After viewing the analyst's problem as an online problem, we may straightforwardly
 establish the remaining correspondences. First, note that the simplex $\Theta$ is convex
 and bounded. Second, note that we may imagine the adversary in the panel prediction game
 as picking loss functions $\ell_t (\cdot)$ of the form $\theta
 \mapsto (y_{0t} - \theta'\by_t)^2$, parametrized by the potential outcomes $(y_{0t},
 \by_t)$. These loss functions are indeed convex in $\theta$ and bounded,
 since both $\theta$ and $\bY$ are bounded. Finally, note that the average loss
 \eqref{eq:expectedlosspanel} is equal to $\frac1T\sum_{t=1}^T
 \ell_t(\theta_t)$, which is simply the total loss in the online protocol scaled by
 $\frac1T$.\footnote{\label{foot:transpose}\Copy{transpose}{It may be tempting
 to ask whether the same argument applies to
 ``horizontal regression'' \citep{athey2021matrix}, where one regresses $y_{iS}$ on $y_{i1},\ldots, y_
 {iS-1}$,
 perhaps constraining the coefficients to some bounded, convex set.
 Since synthetic control can be viewed as a ``vertical regression,'' where one regresses
 $y_ {0t}$ on $y_{1t},\ldots, y_{Nt}$, it seems we may apply our argument to the
 transposed $\bY$ matrix. Indeed, we may formulate analogous claims by replacing $t$
 with $i$, $s$ with $j$, $S$ with some randomly chosen unit $M \in [N]$, and $T$ with $N$.
 However, a difficulty with this interpretation is that synthetic control
 \eqref{eq:ftl} naturally only uses information in the past ($t<S$), but the analogous
 restriction in horizontal regression, $i < M$, for a randomly chosen treated unit $M \in
 [N]$, is much less natural.}}

 \label{page:ftl1}

 Having recognized our setup as an instance of online convex optimization, the main
 observation of this paper recognizes that synthetic control is an online learning
 algorithm known as \emph{Follow-the-Leader} (FTL). FTL, under 
 \cref{item:oco1,item:oco2,item:oco3}, is the algorithm
 that, when prompted for a decision in \cref{item:oco1},
 simply chooses $\theta_t$ to minimize past
 losses:\footnote{FTL is also known as fictitious play in game theory
 \citep{brown1951iterative}. The name
 ``follow-the-leader,'' coined by \cite{kalai2005efficient}, is popular in the recent
 computer science literature. For an introduction to FTL
 and similar algorithms, see Chapter 5 in \citet{hazan2019introduction} and Chapters 1 and
 7 in \citet{orabona2019modern}.}\footnote{\label{ft:unique}\Copy{ftlunique}{When there
 are multiple minima, the choice of
  $\theta_t$ does not affect our theoretical guarantees. Nevertheless, it seems sensible in
  practice to take the minimum that is smallest in some norm, e.g. $\norm{\cdot}_2$.}} \[
\theta_t \in \argmin_{\theta \in \Theta} \sum_{s < t} \ell_s(\theta).
\]
 }
 
 \label{page:ftl}

\begin{obs}
\label{obs:synth} Synthetic control \eqref{eq:ftl} is an instance of FTL applied to the
panel prediction protocol \cref{item:p1,item:p2,item:p3}. 
\end{obs}

Standard online convex optimization results on \emph{regret} then apply to synthetic
control as well. Before introducing these results, let us define regret as the gap
between the total loss of a strategy $\sigma$ and the best fixed weights $\theta$ in
hindsight: 
\begin{align}
	\reg_T(\sigma; \bY) &\equiv \sum_{t=1}^T \ell_t(\theta_t) - \min_{\theta \in \Theta}
   \sum_
	{t=1}^T \ell_t(\theta) \\ &= {\sum_{S=1}^T (y_{0S} -
   \theta_S'\by_S)^2 - \min_
    {\theta
	\in \Theta} \sum_
	{S=1}^T (y_{0S} -\theta' \by_S)^2} \label{eq:reg2}\\ 
    &= T \pr{\E_S[(y_{0S} - \theta_S'\by_S )^2] - \min_{\theta \in
    \Theta}
    \E_S
    [(y_{0S} - \theta'\by_S )^2]} \label{eq:reg3} \\
    &\ge T \pr{\E_S[(y_{0S} - \theta_S'\by_S )^2] - 
    \E_S
    [(y_{0S} - \theta'\by_S )^2]} \text{ for any $\theta \in \Theta$}.
    \label{eq:reg4}
\end{align} 
\eqref{eq:reg2} observes that, in our setting, regret is the
 difference between total squared prediction error of a strategy $\sigma$ and that
 of the best fixed weights $\theta$ chosen in hindsight, summing over
 hypothetical treatment times $S$.
\eqref{eq:reg3} interprets the sum of losses as $T$ times the expected loss under random
 treatment timing. Finally, \eqref{eq:reg4} observes that regret is an upper bound of
 the expected error gap between the strategy $\sigma$ and any fixed weights
 $\theta$. We refer to $\argmin_{\theta\in\Theta} \sum_
    {S=1}^T (y_{0S} - \theta'\by_S)^2$ as the \emph{oracle weighted match}---the best set
    of weights for a given realization of the data $\bY$.

Focusing on regret rather than loss shifts the goalposts from performance to
\emph{competition}, which is a more fruitful perspective in our adversarial setting. After
all, we cannot hope to obtain meaningful loss control as the all-powerful adversary can
make the analyst miserable. However, the crucial insight of regret analysis is that, for
certain strategies $\sigma$, the adversary cannot simultaneously make the analyst suffer
high loss while letting some fixed strategy $\theta$ perform well---in other words, if any
fixed $\theta$ performs well, then $\sigma$ performs almost as well over time. Indeed, if
regret is sublinear, i.e., $\reg_T \le o(T)$,\footnote {We mean $\reg_T \le o(T)$ in the
sense that $\limsup_{T\to\infty}\frac{1}{T} \reg_T \le 0$, since it is possible for
$\reg_T$ to be negative. Following the online convex optimization literature, we sometimes
refer to $\sigma$ as no-regret if it has sublinear regret. } then the strategy $\sigma$
never performs much worse than any fixed weights $\theta$, on average over hypothetical
treatment timing $S$. In this case, we can interpret $\sigma$ as a strategy that is
\emph{competitive} against the class of weighted matching estimators.

It may seem surprising that these no-regret strategies $\sigma$ exist in the first place.
We emphasize that $\sigma$ can output different weights $\theta_t$, chosen adaptively over
time, while $\sigma$ is compared to an oracle that uses the best fixed weights. As a
result, $\sigma$ can compensate for its lack of oracle access by changing its choices
judiciously over time.

 The main result of this paper shows that the regret of synthetic control under quadratic
 loss is logarithmic in $T$. The result follows from a direct application of \cite
 {hazan2007logarithmic}'s regret bound for FTL (Theorem 5 in their paper, reproduced as
 \cref {thm:hazan} in the appendix).
\begin{theorem}
\label{thm:ftlregret}
With bounded outcomes $\norm{\bY}_\infty \le 1$, 
synthetic control \eqref{eq:ftl}, denoted $\sigma$, satisfies the regret
bound\footnote{We say $f(N,T) = O(g(N,T))$ for $g(N,T) > 0$ if, for any sequence $N_T < T$
and $T \to \infty$, \[
\limsup_{T \to \infty} \, \frac{f(N_T,T)}{g(N_T,T)} < \infty.
\]
In the conclusion of \cref{thm:ftlregret}, the inequality does not require $T > N$. The
assumption $T >N$ is only used for the simplification $16 N (\log (\sqrt{N}T) + 1) = O
(N\log N + N\log T) = O(N\log T)$. Of course, the regret bound is less
interesting if $\limsup N \log T / T > 0.$} \[
\reg_T(\sigma, \bY) \le 16 N (\log (\sqrt{N}T) + 1) = O(N \log T).
    \]

\end{theorem}

\Cref{thm:ftlregret} shows that the synthetic control strategy \eqref
 {eq:ftl} achieves logarithmic regret---and as a result, the average difference between
 the losses of synthetic control and losses of the oracle weighted match vanishes quickly
 as a function of $T$.\footnote{Restricting $\theta$ to the simplex $\Theta$---a debated
 choice in the synthetic control literature---is somewhat important for the dependence
 on $N$, in so far as the simplex is bounded in $\norm{\cdot}_1$. This is a consequence
 of the assumption that the outcomes $\bY$ are bounded in the dual norm $\norm
 {\cdot}_\infty$, which implies a bound on $\theta'\by_t$ that is free of $N,T$. In
 contrast, if we let $\Theta = \{\theta : \norm{\theta}_2 \le D/2\}$ be an
 $\ell_2$-ball, then the regret bound worsens to $O(D^2N^2\log (T))$. } In particular,
 if there exists a weighted average of the untreated units' outcomes that tracks $\by(0)$
 well, then the average one-step-ahead loss of synthetic control estimates is
 only worse by $O\pr{\frac{N\log T}{T}}$. 

On its own, \cref{thm:ftlregret} is purely an optimization result; we now offer a few
comments on its statistical implications. As a preview, under random treatment timing,
\cref{thm:ftlregret} implies that the \emph{risk} of estimating the causal effect at time
$S$ for synthetic control is not too much higher than that for any weighted matching
estimator. Indeed, if any weighted matching estimator performs well, then synthetic
control achieves low risk as well. Our discussion below translates \cref{thm:ftlregret}
into guarantees on the expected loss at treatment time---expressing regret as
\eqref{eq:reg3}---which relies on the design assumption that $S$ is randomly assigned.
Nevertheless, we stress that we could view \cref{thm:ftlregret} purely as guarantees of
average loss over hypothetical timings $S$---expressing regret only as
\eqref{eq:reg2}---which does not require a treatment timing assumption.

We can  interpret regret as a gap in the design-based {risk} of
estimating treatment effects. Specifically, we can interpret the expected loss of
predicting the untreated outcome as the risk of estimating the treatment effect: 
\begin{align}
\risk(\sigma, \bY, \by(1)) &\equiv \E_S\bk{
    (\TE_S - \hat \TE_S(\sigma))^2
} \nonumber \\&\equiv \E_S\bk{
    ((y_S(1) - y_{0S}) - (y_S(1) - \hat y_{S}))^2 
} \nonumber \\ &= \E_S[(y_{0S} - \hat y_{S})^2]. \label{eq:riskexpect}
\end{align}
Hence, \eqref{eq:reg3} and \eqref{eq:riskexpect}, combined with \cref{thm:ftlregret},
imply that the risk
of using
synthetic control is no more than
$N\log T/T$ worse than the risk of the oracle weighted match,\footnote{We
slightly abuse notation and use $\theta$ to denote the strategy that
outputs $\theta$ every period.} regardless
of the potential outcomes $\bY, \by(1)$:
\begin{equation}
    \risk(\sigma, \bY, \by(1)) - \min_{\theta\in \Theta} \risk(\theta, \bY,
    \by(1)) = \frac{1}{T}\reg_T
(\sigma, \bY) = O\pr{\frac{N\log T}{T}}.
\label{eq:riskbound}
\end{equation}
This observation connects regret on prediction of the untreated potential outcome with
differences in the
risk of estimating treatment effects. Roughly speaking, \eqref{eq:riskbound} shows that
synthetic control estimates of one-step-ahead causal effects are
competitive against that of any fixed weighted match, for any realization
of $\bY, \by(1)$,
on average over $S$.

Of course, since the guarantee \eqref{eq:riskbound} holds for every $\bY$, it continues
to hold when we average over $\bY$ and $\by(1)$, over a joint distribution $P$ that
respects
the boundedness condition $\norm{\bY}_\infty \le 1$. In this sense, analyzing regret in
the adversarial framework not only does not preclude statistical
interpretations, but rather {facilitates} analysis in a wide range of outcome
models.\footnote{The technique of ``online-to-batch conversion'' in the online learning
literature exploits this intuition to prove results in batch (i.i.d.) settings via
results in online adversarial settings. }
Formally, let $\mathcal P$ be a family of distributions for $\bY, \by(1)$ such that $P
(\norm{\bY}_\infty
\le 1) = 1$ for all $P \in \mathcal P$. 
Under an outcome model $P$, we may understand $\risk(\sigma, \bY, \by(1))$ as
\emph{conditional risk} and $\E_P \risk(\sigma, \bY, \by(1))$ as 
\emph{unconditional risk}.
 Then,
\eqref{eq:riskbound} implies that\footnote{\cite{abernethy2009stochastic} show that a
minimax theorem applies, and \[
\sup_P \inf_\sigma \E_P\bk{\risk(\sigma, \bY, \by(1))] - 
\min_{\theta\in \Theta} \risk(\theta, \bY, \by(1))} = \frac{1}{T}
\inf_\sigma \sup_\bY \reg_T(\sigma, \bY).
\]
Note that the $\le$ direction is immediate via the min-max inequality. This result shows
that the worst-case optimal risk differences in a stochastic setting (i.e. the analyst
knows $P$ and  responds to it optimally) is equal
to minimax regret. In this sense, worst-case regret analysis is not by itself
conservative for a stochastic setting---minimax regret is a tight upper bound for
performance in stochastic settings.
} 
\begin{equation}
   \sup_{P \in \mathcal P} \E_P\bk{\risk(\sigma, \bY, \by(1)) - 
\min_{\theta\in \Theta} \risk(\theta, \bY, \by(1))} = O
\pr{\frac{N\log T}{T}}.
\label{eq:unconditionalrisk} 
\end{equation}
Therefore, the unconditional risk of synthetic control is never much worse than
the risk of the oracle weighted match \[R_\Theta^* \equiv \E_P\bk{\min_
{\theta\in \Theta}
\risk(\theta,
 \bY, \by
 (1))}.\] Hence, if the data-generating process $P$ guarantees that $R_\Theta^*$ is small,
then synthetic control achieves low expected risk as well. Concretely speaking, this
latter requirement is that, for most realizations of the data, had we observed all the
potential outcomes, we could find a weighted match that tracks the potential outcomes
$y_{01},\ldots, y_{0T}$ well, so that\footnote{Also, observe that $
\E_P[\min_{\theta \in \Theta} \frac{1}{T}\sum_{t=1}^T (y_{0t} - \theta' \by_t)^2]
\le \min_{\theta \in \Theta} \E_P[\frac{1}{T}\sum_{t=1}^T (y_
 {0t} - \theta' \by_t)^2], $ and thus the guarantee \eqref{eq:unconditionalrisk} is
 stronger in the sense that it allows the oracle $\theta$ to depend on the realization
 of the data.} \[
\E_P\bk{\min_{\theta \in \Theta} \frac{1}{T}\sum_{t=1}^T (y_{0t} - \theta' \by_t)^2}
\approx 0.
 \]

 In many empirical settings, it seems plausible that the oracle weighted match performs
  well.\footnote {We recognize that under many data-generating models, there is
  unforecastable, idiosyncratic randomness in $y_{0t}$. As a result, there may not exist a
  synthetic match that perfectly tracks the \emph{realized} series $y_{0t}$ (even though
  such a match may exist that tracks various conditional expectations of $y_{0t}$ quite
  well). In many such cases, since squared error can be orthogonally decomposed, risk
  differences for estimating $y_ {0t}$ are also risk differences for estimating
  conditional means $\mu_{t}$ of $y_ {0t}$. We discuss these results in \cref{asec:means}.
  } \cite{abadie2021using} states the following intuition in many comparative case
  studies: ``[T]he effect of an intervention can be inferred by comparing the evolution of
  the outcome variables of interest between the unit exposed to treatment and a group of
  units that are similar to the exposed unit but were not affected by the treatment.''
  More formally speaking, a well-fitting oracle weighted match also resembles---and
  implies---\cite{abadie2010synthetic}'s assumption that there exists a perfect
  pre-treatment fit of the outcomes. When the oracle weighted match performs well, our
  regret guarantees imply a guarantee on the loss of the feasible synthetic control
  estimator, making it an attractive option for causal inference in comparative case
  studies.

   Even if no weighted average of the untreated units tracks $y_{0t}$ closely,
   synthetic control continues to enjoy the assurance that it performs almost as well
   as the best weighted match. Moreover, in the general online learning setup \cref{item:oco1,item:oco2,item:oco3}, this
   no-regret property cannot be
   attained
   without choosing $\theta_t$ in some data-dependent manner.\footnote
   {See \cref{asub:nofixed} for a simple argument in a general setup with unspecified
   $\ell(\cdot)$. Since simple DID does not choose weights adaptively, it fails to
   control regret against the class of weighted DID estimators that we discuss in 
   \cref{sec:extension}.} This
   observation rules out alternatives such as simple difference-in-differences, which
   does not aggregate the control units in a data-dependent manner. In contrast, in \cref
   {sec:extension}, we additionally show that synthetic control on differenced data
   performs almost as well as the best \emph{weighted} difference-in-differences
   estimator, a popular class of estimators in practice.

\section{Extensions}
\label{sec:extension}

   \subsection{Non-uniform treatment timing}
   \label{sub:nonuniform}

  The previous interpretations---in \eqref{eq:reg3} and \eqref{eq:riskexpect}---rely on
  interpreting average loss over hypothetical values of $S$ as expected loss over $S$,
  which requires uniform treatment timing $S \sim \Unif[T]$. Despite being plausible in
  certain settings and appearing elsewhere in the literature \citep
  {doudchenko2016balancing,bottmer2021design}, this assumption is perhaps
  crude.\footnote{\cite{doudchenko2016balancing} discuss inference in synthetic control
  via randomization of the treatment timing in their Section 6.2. \cite{bottmer2021design}
  consider randomization of the treated period in their Assumption 2, though, in their
  setting, the treatment lasts only one period. We also note that the randomness \emph{per
  se} of $S$ conditional on $\bY$ can be realistic, but that its distribution is uniform
  and known is restrictive.} To some extent this is inevitable: Since we are agnostic on
  the outcome generation process, it is unavoidable to make treatment timing assumptions
  in order to obtain nontrivial statistical results on estimation of causal quantities.
  Nevertheless, note that such an assumption is only necessary for interpreting average
  losses as expected losses. The \emph{a priori} proposition that \emph{it is reasonable
  to expect a causal estimator to predict well relative to some oracle, at least on
  average over hypothetical treatment timings,} strikes us as defensible. Accepting this
  dictum relieves us of any need to model treatment timing.

  Even if we wish to maintain the interpretation of average loss as expected loss, we can
  relax the uniform treatment timing assumption. In this subsection, we show that if
  the treatment timing distribution is known, then a weighted version of synthetic
  control achieves logarithmic weighted regret. Moreover, even if the treatment timing
  distribution is non-uniform, unknown, and possibly chosen by the adversary, we
  continue to show that synthetic control performs well if some weighted
  average of untreated units predicts $y_{0S}$ accurately. Both results have constants
  that worsen if the treatment timing distribution deviates far from $\Unif[T]$.

  Suppose the conditional distribution $
  (S \mid \bY)$ is denoted by $\pi =(\pi_1,\ldots,\pi_T)'$, which may depend on
  $\bY$. Note that, for a known $\pi$, we may apply the
  same argument in \cref{thm:ftlregret} to the following weighted synthetic control
  estimator:
  \begin{equation}
      \hat\theta_S^\pi \in \argmin_{\theta \in \Theta} \sum_{t < S} \pi_t (y_{0t} -
\theta'\by_t)^2,
\label{eq:weightedsynth}
  \end{equation}
  by redefining the loss functions $\ell_t(\cdot)$.
  This argument shows that \eqref{eq:weightedsynth} achieves $\log T$
  weighted regret, stated in the following corollary. Note that \eqref{eq:weightedsynth}
  implements FTL with
  loss functions $\ell_t(\theta) \equiv \pi_t (y_{0t} - \theta'\by_t)^2$,
  and hence the argument of \cite{hazan2007logarithmic} applies.
    
\begin{cor}
\label{cor:knownweights}
      Suppose $S \sim \pi$, $\frac{1}{CT}\le \pi_t \le \frac{C}{T}$ for some $C$,
      and $\norm{\bY}_\infty \le 1$. 
      Then weighted synthetic control \eqref{eq:weightedsynth}, denoted $\sigma_\pi$,
      achieves weighted regret bound \begin{align}
      \reg_T(\sigma_{\pi}; \pi, \bY) &\equiv T \cdot \pr{ \E_{S\sim \pi} [
      (y_
      {0S} -
      \hat\theta_S'
      \by_S)^2] - \min_{\theta \in \Theta } \E_{S\sim \pi} [(y_{0S} - \theta'
      \by_S)^2]} \label{eq:weightedregret}
      \\&\le 16C^3 N \bk{\log{\frac{\sqrt{N}T}{C^2}} + 1} = O(C^3 N \log T). \nonumber
      \end{align}
  \end{cor}
\Cref{cor:knownweights} shows that the weighted regret---a difference in $\pi$-expected
 loss---is logarithmic in $T$, thereby controlling the worst-case gap between weighted
 synthetic control and the oracle weighted match for the expected loss. Assuming a
 known $\pi$ could be reasonable. With a known dynamic treatment regime, $\pi$ can depend
 on $\bY_{1:S-1}$, but is known whenever the analyst is prompted for a prediction at time
 $S$.\footnote{\label{foot:dynamictrx}Since the bound is for a fixed $\bY$, we can allow
 $\pi$ to depend on $\bY$, so long as $\pi_t(\bY)$ is known at time $t+1$ so that the
 analyst can compute \eqref{eq:weightedsynth}. This allows for \cref{cor:knownweights} to
 be applied in the following example, which is a more realistic design-based setting.
 There is a known \emph{dynamic treatment regime} \citep{chakraborty2014dynamic}
 parametrizing the treatment hazard: That is, \[
 \P(S = t \mid S \ge t, \bY) = r_t(\bY_{1:t-1})
 \]
 for some known $r_t(\cdot)$. Then $
\pi_t(\bY) = \P(S = t \mid \bY) = (1-r_1)\cdots(1-r_{t-1}) r_t
 $
is a function of $\bY_{1:t-1}$. We thank Davide Viviano for suggesting this extension. }
 We can also interpret \cref{cor:knownweights} as providing guarantees on differences in
 Bayes risk under the analyst's prior $S \sim \pi$, independent of $\bY$. %

Even when $\pi$ is \emph{unknown} and chosen by the adversary, we can
bound
the loss of
unweighted synthetic control, so long as $\pi$ is not too far from uniform.

\begin{cor}
\label{cor:weightedloss}
Suppose $S \sim \pi$, $\pi_t \le C/T$ for some $C$, and $\norm{\bY}_\infty \le 1$.
Then synthetic control \eqref{eq:ftl}, denoted $\sigma$, achieves the following bound on 
\emph{the
expected loss} \begin{equation}
    \E_{S \sim \pi} \bk{ (y_{0S} - \hat \theta_{S}'\by_S)^2 } \le C\pr{\min_
{\theta \in \Theta} \frac{1}{T} \sum_{t=1}^T (y_{0t} - \theta'\by_t)^2 +
\frac{1}{T} \reg_T(\sigma; \bY)},
\label{eq:weightedriskpointwise}
\end{equation}
where $\reg_T(\sigma; \bY)$ is defined by \eqref{eq:reg2}.
Hence, for any joint distribution $Q$ of $(\bY, S)$ where $Q(S = t \mid
\bY)
\le C/T$ for all $t$, and $Q(\norm{\bY}_\infty \le 1) = 1$, we have the average loss
bound
\begin{equation}
    \E_Q[(y_{0S} - \hat \theta_{S}'\by_S)^2] \le C \pr{\E_Q\bk{\min_
{\theta \in \Theta} \frac{1}{T} \sum_{t=1}^T (y_{0t} - \theta'\by_t)^2} + O\pr{
\frac{N\log T}{T}}}.
\label{eq:weightedriskaverage}
\end{equation}
\end{cor} 

The result \eqref{eq:weightedriskpointwise} shows that, uniformly over all
 bounded $\bY$ and bounded treatment distributions $\pi$, the expected squared error is bounded by the average loss of the oracle weighted match plus the
 regret, all scaled with a constant $C$ that indexes how far $\pi$ deviates from the
 uniform distribution. Under the  same assumption that the oracle weighted match
 performs well on average, \eqref{eq:weightedriskpointwise} continues to show that the
 treatment
 estimation risk of synthetic control is small.
Since such a result is valid for all $\bY$ and $\pi$, we may understand 
\eqref{eq:weightedriskpointwise} as a bound that holds even in a setting where the
adversary picks both $\bY$ and $\pi$, with the restriction that $\pi_t \le C/T$, but
otherwise unrestricted in the dependence between $\bY$ and $\pi$.

As before, since \eqref{eq:weightedriskpointwise} is a guarantee uniformly over $\bY$,
it is also a guarantee when we average over $\bY$ under an outcome model, yielding
\eqref{eq:weightedriskaverage}. Again, \eqref{eq:weightedriskaverage} shows that \emph{for any}
joint distribution of the bounded outcomes and the treatment timing, the unconditional
risk of
synthetic control is small when the expected oracle conditional risk, $\E_Q[\min_
{\theta \in \Theta} \frac{1}{T} \sum_{t=1}^T (y_{0t} - \theta'\by_t)^2]$,  is
small---so long as $S$ has sufficient
randomness conditional on $\bY$ so that $C$ is not too large.

So far, we have considered weighted averages of untreated units as the class of competing
estimators. These competing estimators are matching estimators. However, a more common
class of competing estimators in applications are difference-in-differences (DID)
estimators. It turns out that synthetic control on preprocessed data has regret
guarantees against a class of DID estimators, which we turn to in the next subsection.

\subsection{Competing against DID}

\Cref{sec:theory} shows that the original synthetic control estimator is competitive
against a class of matching estimators that use weighted averages of untreated units as
matches for the treated unit. However, in many applications in economics, matching
estimators are much less popular than DID estimators, since the latter accounts for
unobserved confounders that are additive and constant over time. In this subsection, we
show that synthetic control on differenced data is competitive against a large class of
DID estimators. Additionally, \cref{asub:otherDID} offers regret guarantees against
other flavors of DID estimators.

In practice, a common DID specification is the following
two-way fixed effects
regression: \[
\min_{\mu_i, \alpha_t, \lambda} \sum_{i=0}^N \sum_{t=1}^S \pr{y_{it}^{\text{obs}} -
\mu_i -
\alpha_t - \lambda \one\bk{(i,t) = (0, S)}}^2,
\]
where the observed outcome $y_{it}^{\text{obs}} = y_{it}$ for all $(i,t) \neq (0,S)$, and
$y_{0S}^{\text{obs}}
= y_S(1)$.
This specification regresses the observed outcomes on unit and time fixed effects, and
uses the estimated coefficient $\lambda$ as an estimate of the treatment effect $y_
{S}(1) - y_{0S}$. Implicitly, this regression uses the estimated fixed effects
$\mu_0 +\alpha_S$ as a forecast for the unobserved $y_{0S}$. 
\Copy{sdid1}{
We consider a weighted generalization of this regression, a special case of the synthetic
DID estimators in \cite{arkhangelsky2021synthetic}:\footnote{
\label{foot:sdid}The weight $w_0$ does
not affect $\mu_0 +\alpha_S$ achieving the optimum in the least-squares problem, per
the calculation in \cref{asub:twfe}. As a result, we normalize $w_0 = 1$. Moreover,
specifically, \eqref{eq:twfe}
is a special case of synthetic DID, (1) in \cite{arkhangelsky2021synthetic}, with only
unit-level weights and no
time-level
weights. }\footnote{\eqref{eq:twfe} is underdetermined if $S = 1$. The ensuing discussion
assumes $\sum_{i=1}^N w_i y_{i1}$ is the weighted two-way fixed effects prediction for $y_
{01}$.}
\begin{equation}
 \min_{\mu_i, \alpha_t, \lambda} \sum_{i=0}^N  \sum_{t=1}^S w_i (y_{it}^{\text{obs}} -
\mu_i -
\alpha_t - \lambda \one\bk{(i,t) = (0,S)})^2 \quad w_0 = 1, \sum_{i=1}^N w_i = 1, w_i
\ge 0.
\label{eq:twfe}   
\end{equation}
} For convex weights $w = (w_1,\ldots, w_N)'$, denote by $\sigma_{
\mathrm{TWFE}}(w)$ the
strategy
that estimates \eqref{eq:twfe} on the data $(\bY_{1:t-1}, \by_t)$ at time $t$,\footnote{The value of $y_{0t}$ does not enter $\alpha_S + \mu_0$ since it is
absorbed by the coefficient $\lambda$.} and outputs the estimated coefficients $\mu_0 +
\alpha_t$ as a prediction for $y_{0t}$. By varying over $w \in \Theta$, we obtain a
class of competing DID strategies, where conventional DID corresponds to picking uniform
weights $w = (1/N,\ldots, 1/N)'$. We calculate in \cref{asub:twfe} that
the
prediction that $\sigma_{\mathrm{TWFE}}(w)$ makes is \[
\hat y_{t}(\sigma_{\mathrm{TWFE}}(w)) = \frac{1}{t-1} \sum_{s=1}^{t-1} y_{0s} + w' \pr{\by_t
- \frac{1}{t-1} \sum_{s=1}^{t-1} \by_s}\qquad t \ge 2,
\]
which simply uses the outcome difference against historical averages of untreated units
to forecast that of unit $0$. Note that this strategy amounts to using a weighted match
with weight $w$
on the \emph{differenced data} \[\tilde y_{i1} = y_{i1}\qquad \tilde y_{it} \equiv y_{it} -
\frac{1}{t-1} \sum_{s=1}^
{t-1} y_
{is} \qquad |\tilde y_{it}| \le 2\] to forecast the same differences of unit $0$, $\tilde
y_{0t}$. Therefore,
we may apply \cref{thm:ftlregret} and show the following regret
bound. 

\begin{theorem}
\label{prop:twfe}
    Consider synthetic control on the differenced data, where the analyst computes \[
\hat \theta_t \in \argmin_{\theta \in \Theta} \sum_{s < t} \pr{\tilde y_{0s} - \theta'
\tilde \by_{s}}^2
    \]
    and predicts $\hat y_t = \frac{1}{t-1}\sum_{s< t} y_{0s} + \hat \theta_t'\tilde
    \by_t$. Here, $\tilde y_{it} = y_{it} - \frac{1}{t-1}\sum_{s< t}y_{is}$ is the
    difference against historical means, and $\tilde \by_t = (\tilde y_{1,t}, \ldots,
    \tilde y_{N,t})'$. Then we have the following regret guarantee against the oracle
    $\sigma_{\mathrm{TWFE}}$, whose weights are chosen ex post:\[
\sum_{t=1}^T (y_{0t} - \hat y_t)^2 - \min_{\theta \in \Theta} \sum_{t=1}^T (y_{0t} -
\hat y_t(\sigma_{\mathrm{TWFE}}(\theta)))^2 \le C N \log T
    \]
    for some constant $C$. 
\end{theorem}

\Cref{prop:twfe} shows that synthetic control on differenced data controls regret against
the class of DID estimators \eqref{eq:twfe}.\footnote{The benchmark class of DID
estimators
in \cref{prop:twfe} output predictions in a sequential manner, in so far as the
coefficients in the regression \eqref{eq:twfe} depend on $S$.
 In contrast,
\cref{prop:staticdid}
compares synthetic control against a class of static DID estimators that do not exhibit
this feature.} In particular, the class of DID benchmarks corresponds to weighted
two-way fixed effects regressions, and
synthetic control is competitive against any fixed weighting. In this sense,
\cref{prop:twfe} builds on the intuition that synthetic control is a generalization of DID
\citep{doudchenko2016balancing} to show that a version of synthetic control performs as
well as any weighted DID estimator. Again, if any weighted DID estimator performs well,
then \cref{prop:twfe} becomes a performance guarantee on synthetic control. Moreover,
since
 \eqref{eq:twfe} is a popular alternative for many practitioners---setting
 aside whether there is a weighted DID that performs well---\cref{prop:twfe} shows that
 it is without much loss to use synthetic control in such settings
 instead. 
 \Copy{recommendation}{Since DID is more
 popular in practice than weighted matching, competitive performance
 against DID is a more relevant consideration, which suggests
 prioritizing synthetic control on
 differenced data $\tilde y_{it}$ over classic
 synthetic control \eqref{eq:ftl}.\footnote{This comment is with the
 caveat that the constant in \cref{prop:twfe} is worse than that in \cref{thm:ftlregret}.
 It seems possible to
 further improve
 the guarantee in \cref{prop:twfe}, since in our proof, we solely use the
 implication 
 $|\tilde y_{it}| \le 2$ and do not restrict
 the adversary from choosing $\tilde y_{it}$ where the implied $|y_
 {it}|>1$. We leave such a refinement to future work. 
 
 Of course, this observation also implies that \cref{prop:twfe} holds
 without bounded outcomes $
 \norm{\bY}_\infty
 \le 1$  and solely with bounded differences $\max_{i, t} |\tilde y_{it}|
 \le 2$.
 }}

\Copy{sdid}{
To the best of our knowledge, the difference scheme  $\tilde y_{it}$ has yet to be
 considered in the literature. We do note that since the resulting
 predictions are
 equivalent to a weighted two-way fixed effects regression, this proposed synthetic
 control
 scheme can be thought of as synthetic DID \citep{arkhangelsky2021synthetic} with weights
 chosen by constrained least-squares on $\tilde y_{it}$. } We also note
 that $\tilde y_
 {it}$ is slightly different from \cite{ferman2021synthetic}'s demeaned synthetic control,
 which takes the difference $\dot y_{it} \equiv y_{it} - \frac{1}{t} \sum_{s=1}^{t}
 y_{is}$. In \cref {asub:otherDID}, we show that 
 \cite{ferman2021synthetic}'s demeaned
 synthetic control achieves logarithmic regret against a different class of DID estimators
 that we call static DID estimators.\footnote{Under certain conditions, \cite{ferman2021synthetic} (Proposition 3)
 show that the demeaned synthetic control in \cref{prop:staticdid} dominates DID with
 uniform weighting $\theta_i = 1/N$. The results \cref{prop:staticdid,prop:twfe} are in
 a similar flavor, and show that synthetic control is competitive against DID with any
 fixed weighting, on average over random assignment of treatment time. Of course,
 \cref{prop:staticdid,prop:twfe} are not generalizations of \cite{ferman2021synthetic}'s
 result---for one, we consider average loss under random treatment timing, and
 \cite{ferman2021synthetic} consider a fixed treatment time under an outcome model, with
 the number of pre-treatment periods tending to infinity.}  Another popular alternative  is first-differencing
 \citep{abadie2021using}, which by similar arguments may be shown to control regret
 against a class of \emph{two-period} weighted DID strategies that output $
\hat y_{t} (\sigma_{\text{2P-DID}}(\theta)) \equiv y_{0 t-1} + \theta'\pr{\by_t - \by_
{t-1}}$ as successive predictions.

\subsection{Regularization, covariates, and other extensions}
\label{sub:reg}

\Cref{thm:ftlregret} shows that synthetic control, as FTL, gives logarithmic regret when
we consider quadratic loss. However, to some extent this bound is an artifact of using
squared losses, whose curvature ensures that the FTL predictions do not move around
excessively over time. If we replace the loss function with the absolute loss $|\hat y -
y|$, then the regret may be linear in $T$---no better than that of the trivial prediction
$\hat y_t \equiv 0$ \citep[see Example 2.10 in][]{orabona2019modern}.

Motivated by the lack of general sublinear regret guarantees in FTL, the online learning
literature proposes a large class of algorithms called \emph
{Follow-The-Regularized-Leader} (FTRL), where regularization helps stabilize the FTL
predictions. With linear prediction functions $f(\by; \theta) = \theta'\by$, such
strategies take the form
\begin{equation}
\theta_t \in \argmin_{\theta \in \Theta} \sum_{s<t} \ell(\theta'\by_s, y_{0s}) + 
\frac{1}{\eta} \Phi(\theta)
    \label{eq:ftrl}
\end{equation}
for some convex penalty $\Phi$ and regularization strength $1/\eta > 0$. Here, we let
$\ell
(\cdot, \cdot)$ denote a
generic
convex and bounded loss function, generalizing our previous framework. Many regularized
variants
of synthetic control have been proposed 
\citep[among others,][]
 {chernozhukov2021exact,doudchenko2016balancing,hirshberg2021least}. These regularized
 estimators have the form \eqref{eq:ftrl}, though most such estimators are based on
 quadratic loss.

\begin{obs}
    Regularized synthetic control with penalty $\Phi(\cdot)$ is an instance of FTRL, where
    $\ell(\cdot, \cdot)$ is typically quadratic loss.
\end{obs}

\Copy{cov}{ Moreover, we can think of synthetic control with covariates as regularized
synthetic
 control as well. With time-invariant covariates $\bx_j = (x_{1j},\ldots,x_{Nj})'$ for
 $j=1,\ldots, J$, synthetic control may choose weights $\theta$ to
 additionally match the
 covariates \citep[see, e.g., (7) in][]{abadie2021using}: 
 \begin{equation}
 \hat \theta_{S,x} \in \argmin_{\theta \in \Theta} \sum_{t<S} (y_{0t}-\theta'\by_t)^2 +
 \frac{1}{2\eta}\sum_{j=1}^J \eta_j (x_{0j} - \theta'\bx_j)^2, \label{eq:cov_reg}
 \end{equation}
 for some given $\eta_j$ that indexes the importance of matching covariate $j$. Observe
 that, for fixed $x_{0j}, \bx_{j}$,
 \eqref{eq:cov_reg} is a special case of \eqref{eq:ftrl}; in particular, 
 \eqref{eq:cov_reg}
 uses a quadratic penalty of the form \[\Phi
 (\theta) = \frac{1}{2}(\bx-\bX\theta)'H(\bx-\bX\theta)\] for
 some positive definite $H$, vector $\bx$, and conformable matrix $\bX$. Thus, under the
 assumption that the covariates $x_{0j}, \bx_j$ are fixed and not chosen by the adversary,
 we may analyze synthetic control with time-invariant covariates as a special case of FTRL.}

Motivated by the importance of loss function curvature, we slightly generalize and
consider regularized synthetic control estimators using generic loss functions. A
standard result in online convex optimization (e.g. Corollary 7.9 in \cite
{orabona2019modern}, Theorem 5.2 in \cite{hazan2019introduction}) shows that choices of
$\eta$ exist to obtain $\sqrt{T}$ regret.\footnote{This rate matches the lower bound
for linear losses. See Chapter 5 of \cite{orabona2019modern}.} The conditions for this
result are highly general, explaining the popularity of FTRL in online convex
optimization. We specialize to a few choices of the penalty function $\Phi$ in the
synthetic control setting; see \cref{thm:ftrl} for a general statement. 

\begin{theorem}
\label{cor:regularized}
Consider regularized synthetic
control \eqref{eq:ftrl}, equivalently FTRL, with penalty function
$\Phi(\theta)$ and $\theta$ restricted to the simplex $\Theta$. Let $\ell(\theta'\by_t,
y_{0t})$ be a convex loss function in
$\theta$, not necessarily quadratic, to be specified. 
\begin{enumerate}[wide]
    \item Consider the quadratic penalty $\Phi(\theta) = \frac{1}{2} (\bx-\bX\theta)'H
    (\bx-\bX\theta)$. Assume the Hessian $\nabla_{\theta\theta'}\Phi(\cdot) = \bX'H\bX$
    is positive definite with minimum eigenvalue
    normalized to 1. Let $K=\sup_{\theta\in\Theta}  \Phi(\theta) - \inf_{\theta\in\Theta}
    \Phi(\theta)$ be the range of $\Phi(\cdot)$. Then, for both
    squared loss $\ell(\hat y, y) = \frac{1}{2}(y-\hat y)^2$ and linear
    loss $\ell(\hat y, y) =|y-\hat y|$, we have $
\reg_T \le 2\sqrt{2KNT}
    $
    with the choice $\eta = \sqrt{K(2NT)^{-1}}$. 

    Moreover, if $\bx = 0 $ and $ \bX = H = I$,
    then $\Phi(\theta) = \frac{1}{2}\norm{\theta}^2$ is the ridge penalty, for which we
    obtain $\reg_T \le 2\sqrt{NT}$ with the choice $\eta = 1/\sqrt{4NT}$.

    \item For the entropy penalty $\Phi(\theta) = \sum_i \theta_i \log \theta_i + \log(N)$, for both
    squared and linear losses, we have $
\reg_T \le 3\sqrt{T\log N}
    $
    with the choice $\eta = \sqrt{(\log N)/ T}$.
\end{enumerate}
These results hold for any $N,T > 0$ and allow for $T \le N$.
\end{theorem}
\Copy{reg}{
Naturally, these choices correspond to regularized variants of synthetic control. As we
discuss above, quadratic penalties generalize ridge penalization 
\citep{hirshberg2021least}
and matching on covariates.\footnote{Ridge penalties are a special case of elastic net
penalties
proposed by \citep{doudchenko2016balancing}. \cref{thm:ftrl} applies to elastic net penalties
with nonzero $\ell_2$ component as well.

Note that when $\bX \in \R^{J \times N}$ represents pre-treatment covariates of the
control units, $\bX'H\bX$
being positive definite requires that the dimension of the covariates is at least the
number of control units. } The entropy penalty, which is very natural when the parameters
lie on the simplex, is a special case of the proposal in \cite{robbins2017framework}; the
resulting regret bound has better dependence on $N$ and obtains the no-regret property as
long as $\frac{\log N}{T} \to 0$.\footnote{%
Interestingly, $\ell_1$-penalty \citep[proposed by,e.g.,][] {chernozhukov2021exact}
alone is not strongly convex \citep[See Section 9.1.2 of][]{boyd2004convex}, and
\cref{thm:ftrl} does not apply. However, \cref{thm:ftrl} only contains sufficient
conditions, and so this alone is not a criticism of $\ell_1$-penalty.} For these
guarantees, the choice of $\eta$ does require knowledge on the total number of periods
$T$. This may be relaxed via the ``doubling trick'' (see \cite{shalev2011online},
Section 2.3.1), if we allow for different regularization strengths $\eta_S$ for different
realizations of $S$.} %

 We conclude this section by pointing out a few other extensions. First, another weakening
 of the uniform treatment timing requirement can be achieved by considering the maximal
 regret over subperiods of $[T]$, also known as \emph{adaptive regret}. We show in
 \cref{sec:adaptivereg} that a modification to the synthetic control algorithm---which
 still outputs a weighted average of untreated units---achieves worst subperiod regret of
 order $\log T$. Such a result implies that if \emph{we additionally let the adversary
 pick a subperiod} of length $T'$, and treatment is uniformly randomly assigned on this
 subperiod, then modified synthetic control is at most $\frac{\log T}{T'}$-worse on
 expected loss than the oracle weighted match. Of course, this regret guarantee is
 meaningful only when the subperiod is sufficiently long, i.e., $T'\gg \log T$. Second,
 under a design-based framework on treatment timing, we can test sharp hypotheses of the
 form $H_0 : \by(1) - \by(0) = \mathbf{z}$ by leveraging
 symmetries induced by random treatment timing. We briefly discuss inference in
 \cref{asec:inference}.

\section{Conclusion}

\label{sec:conc}

This paper notes a simple connection between synthetic control methods and online
convex optimization. Synthetic control is an instance of Follow-The-Leader, which are
well-studied strategies in the online learning literature. We present standard regret
bounds for FTL that apply to synthetic control, which have interpretations as bounds for
expected regret under random treatment timing. These regret bounds translate to bounds on
expected risk gap under outcome models and imply that synthetic control is competitive
against a wide class of matching estimators. In cases where some weighted match of
untreated units predict the unobserved potential outcomes, these results show that
synthetic control achieves low expected loss. Moreover, the regret bounds can be adapted
to be regret bounds against difference-in-differences strategies. Lastly, we draw an
analogous connection between regularized synthetic control and
Follow-the-Regularized-Leader, a popular class of strategies in online learning.

\Copy{limit}{
We now point out a few limitations of this paper and directions for future
work. First and foremost, the approach we have taken in this paper is deliberately
pessimistic. Living in fear of an adversary constrained solely by bounded outcomes is
perhaps too paranoid for sound decision-making. For instance, this worst-case
perspective is not particularly
amenable to incorporating covariates, since matching on covariates is
inherently based on the hope that the
covariates are predictive of potential outcomes.
 Further constraining the
adversary
\citep{rakhlin2011online} may be an interesting direction for future research. For
instance, it may be fruitful to consider an adversary with a fixed budget for how much
$y_{0t}, \by_t$ deviate from $y_{0,t-1}, \by_{t-1}$. Constraining the adversary may also
render covariates useful, even in a worst-case framework.
}

It may also be interesting to consider alternative online protocols. So far, we have
considered a thought experiment where, before each step $t$, the analyst only has access
to data $\bY_{1:t-1}$ to output a prediction function. In practice, the analyst typically
does have access to $\by_1,\ldots, \by_T$. Alternative protocols have been considered in
the online learning literature. One example is the Vovk--Azoury--Warmuth forecaster
\citep[See Section 7.10 in][] {orabona2019modern}, where we assume the analyst
additionally has access to $\by_t$ before they are prompted for a prediction at time $t$.
In this case, regularized strategies can also achieve $\log T$ regret. Additionally,
\cite{bartlett2015minimax} consider the fixed design setting in which $\by_1,\ldots,
\by_T$ is fully accessible to the analyst before they are prompted for a prediction.
\cite{bartlett2015minimax} give a simple and explicit minimax regret strategy for online
linear regression, which we may adapt into a synthetic control estimator.

\Copy{multi}{We have only considered regret on one-step-ahead prediction for $y_
{0S}$, but synthetic control estimates are often extrapolated multiple time periods ahead
in practice. In attempting to extend our results to $k$-step-ahead prediction, it is
natural to consider $\check y_{it} = (y_{it},\ldots, y_{i,t + k})$, and to attempt a
similar argument on $\check \bY$. The chief difficulty in doing so is one of delayed
feedback, where the analyst cannot update their time-$S$ decision based on loss from times
$1,\ldots, S-1$. That is, for $k$-step-ahead prediction, the analyst, viewed as an online
player who is prompted for a forecast of $\check y_{0,S} = (y_{0S}, y_{0,S+1}, \ldots, y_
{0,S+k-1})$, does
not have access to their prediction loss for $\check y_{0,S-1} = (y_{0,S-1}, y_{0S},
\ldots, y_{0,S+k-2})$, since $y_{0, S+k-2}$ is not yet observed. As a result, unlike
\cref{item:oco3} in the standard online convex optimization protocol, the analyst does not
have access to $\ell_1 (\cdot),\ldots,
\ell_{S-1} (\cdot)$ when making decisions $\theta_S$---rendering our
results here insufficient. That said, delayed feedback---where
the online player only has knowledge of the loss function after $k$ periods---is studied
in online learning \citep
{weinberger2002delayed,korotin2018aggregating,flaspohler2021online}, and we leave an
exploration to future work.}

\bibliographystyle{ecta-fullname}
\bibliography{main.bib}

\begin{appendix}

\section{Proofs and additional results}
\label{asec:app1}

\subsection{Proofs of \cref{thm:ftlregret,cor:weightedloss,cor:knownweights}}
\label{asub:proofs}

We reproduce Theorem 5 of \cite{hazan2007logarithmic} in our notation.
\begin{theorem}[Theorem 5, \cite{hazan2007logarithmic}]
\label{thm:hazan}
Assume that for all $t$, the function $\ell_t : \Theta \to \R$ can be written as \[
\ell_t(\theta) = g_t(v_t'\theta)
\] 
for a univariate convex function $g_t : \R \to \R$ and some vector $v_t \in \R^n$.
Assume that for some $R, a, b > 0$, we have $\norm{v_t}_2 \le R$ and for all $\theta
\in \Theta$, we have $|g_t'(v_t'\theta)| \le b$ and $g_t''(v_t'\theta) \ge a$, for all
$t$. Then FTL
on
$\ell_t$ satisfies the following regret bound: \[
\reg_T \le \frac{2nb^2}{a} \bk{\log \pr{\frac{DRaT}{b}} + 1}
\]
where $D =  \max_{x,y \in \Theta} \norm{x-y}_2$ is the diameter of $\Theta$.
\end{theorem}

\begin{proof}[Proof of \cref{thm:ftlregret}]
    
    \Copy{uniqueness}{\cref{thm:ftlregret} follows immediately from Theorem 5 in
        \cite{hazan2007logarithmic}, reproduced in our notation as \cref{thm:hazan}. The proof
        of this theorem relies solely on optimality of $\theta_t$ (and the associated
        first-order condition); thus, in the case of multiple minima when minimizing
        $\sum_{t=1}^s \ell_t(\theta)$, any particular sequence of minima $\{\theta_t\}$
        satisfies the guarantee.}

    Since $\Theta$ is
    the simplex, we
     know \[D = \max_{\theta_1,\theta_2\in\Theta} \norm{\theta_1-\theta_2}_2 \le \max_
     {\theta_1,\theta_2\in\Theta} \norm{\theta_1-\theta_2}_1 \le \max_
     {\theta_1,\theta_2\in\Theta} \norm{\theta_1}_1 +
     \norm{\theta_2}_1 = 2.\] We
     choose $g_t(x) = \frac{1}{2}
    (y_{0t} - x)^2$
    with $g_t'(x) = x - y_{0t}$ and $g_t''(x) = 1$. (The scaling by $1/2$ means that
    we obtain a bound on $1/2$ times the regret.)
   The vectors $v_t = \by_t$,
    whose dimensions are $n=N$ and
    whose 2-norms are bounded by $R = \sqrt{N}$. Note that $|\by_t'\theta| = |v_t'\theta|
    \le
    \norm{v_t}_\infty \norm{\theta}_1 \le 1$. Hence, $|g_t'(v_t'\theta)| = |\by_t'\theta -
    y_{0t}| \le |\by_t'\theta| + |y_{0t}| \le 2 \equiv b$
    and $g_t''
    (x) \ge 1 \equiv a$. To summarize, we have $R = \sqrt{N}, a=1, b=2, D=2$, and $n=N$.

    Plugging in, we have \[
\frac{1}{2}\reg_T \le  8N (\log(\sqrt{N}T) + 1),
    \]
    which rearranges into the claim. 
\end{proof}

\begin{proof}[Proof of \cref{cor:knownweights} ]
    The proof for \cref{cor:knownweights} follows similarly, now with \[
g_t(x) = \frac{T}{2} \pi_t (y_{0t}-x)^2 \qquad g_t'(x) = T\pi_t (x-y_{0t}) \qquad g_t''(x)
= T\pi_t.\] Note that,
    since $\frac{1}
{CT} \le \pi_t \le \frac{C}{T}$, we can take $a = 1/C$ and $b = 2C$. Doing so yields the
expression in \cref{cor:knownweights}. 
\end{proof}

\begin{proof}[Proof of \cref{cor:weightedloss}]
    For \cref{cor:weightedloss}, and in particular \eqref{eq:weightedriskpointwise}, by $(1,\infty)$-H\"older's
inequality, \[\sum_{t=1}^T \pi_t
(y_
 {0t} - \hat\theta'_t    \by_t)^2
    \le \pr{\max_t \pi_t} \sum_{t=1}^T (y_{0t}-\hat\theta'_t\by_t)^2 \le \frac{C}{T}
    \sum_{t=1}^T
    (y_{0t}-\hat\theta'_t\by_t)^2.\]
We then apply 
    \cref{thm:ftlregret} to bound $\sum_{t=1}^T (y_{0t}-\hat\theta'_t\by_t)^2 = \min_
{\theta \in \Theta} \sum_{t=1}^T (y_{0t} - \theta' \by_t)^2 + \reg_T.$

\eqref{eq:weightedriskaverage} follows immediately from \eqref{eq:weightedriskpointwise}
by taking the expectation $\E_Q$, noting that 
\begin{align*}
\E_Q[(y_{0S} - \hat\theta_S'y_S)^2] &= \E_Q\bk{\sum_{t=1}^T \one(S = t) (y_
{0t} - \hat\theta_t'y_t)^2} \\ 
&=\E\bk{ \E\bk{\sum_{t=1}^T \one(S = t) (y_
{0t} - \hat\theta_t'y_t)^2 \mid \bY}} \\ 
&= \E\bk{\sum_{t=1}^T Q(S = t \mid \bY) (y_
{0t} - \hat\theta_t'y_t)^2} 
\end{align*}
We then apply \eqref{eq:weightedriskpointwise} to complete the proof. 
\end{proof}

\subsection{Lack of regret control for fixed strategies}

\label{asub:nofixed}

\begin{lemma}
\label{lemma:nofixed}
\Copy{nofixedlemma}{
In the online convex optimization setup, suppose the class of loss functions available to
the adversary satisfies the following property: There exists $\epsilon >0$ such that for
any $\theta \in \Theta$, there exists
$\tilde \theta$ and $\ell_1,\ldots, \ell_T$, for which $\ell_t(\tilde
\theta) \le \ell_t
(\theta) - \epsilon$. Then, the regret of any fixed strategy that outputs $\theta_t =
\theta$ for every period is at least $\epsilon T$. }
\end{lemma}

\Copy{nofixed}{
\begin{proof}
Let $\ell_t, \tilde\theta$ be the sequence of loss functions and alternative satisfying
the required
property on the class of loss functions. Then 
$\reg_T(\theta) \ge \sum_t \ell_t(\theta) - \sum_t \ell_t(\tilde \theta) = \epsilon T$. 
\end{proof}

It is easy to see that the loss functions in the panel prediction problem are rich enough
to satisfy the property in \cref{lemma:nofixed}. Fix, say, $\epsilon < 0.0001$. For any
$\theta$, we can find $\tilde
\theta \in \Theta$ where $\norm{\tilde \theta - \theta}_1 \ge \sqrt{\epsilon}$. Then, there
exists some $\by, \norm{\by}_\infty \le 1$ where \[ |(\tilde \theta - \theta)'\by| =
\max_{\norm{\by}_\infty \le 1}|(\tilde \theta - \theta)'\by| = \norm{\tilde \theta -
\theta}_1 \ge
\sqrt{\epsilon}
\]
since $\norm{\cdot}_1$ is the dual norm to $\norm{\cdot}_\infty$. The adversary chooses
$\by_t = \by$ for all $t \in [T]$ and $y_{0t} = \tilde\theta'\by_t$. Then $\ell_t
(\tilde\theta) = 0$ but $\ell_t(\theta) \ge (\sqrt\epsilon)^2 = \epsilon$.
}
\subsection{Static DID regret control}
\label{asub:otherDID}

We could consider
affine predictors with bounded intercepts\[
f(\by_t ; \theta_0, \theta_1) = \theta_0 + \theta_1'\by_t \quad \Theta = [-2,2] \times
\Delta^{N-1}.
\]
This choice corresponds to variations of synthetic control proposed
by \cite{doudchenko2016balancing} and \cite{ferman2021synthetic} in efforts to mimic
behavior of DID estimators.\footnote{Synthetic control with an
 intercept is equivalent to synthetic control with demeaned data $\{y_s - \frac{1}
 {t}\sum_{k\le t} y_k : s = 1,\ldots, t\}$ \citep
 {ferman2021synthetic}, since the constraint that $\theta_0 \in [-2,2]$ does not bind.}
Our regret bound from \cref{thm:ftlregret} generalizes immediately to the affine
predictions, where the benchmark oracle the regret measures against is \begin{equation}
    \min_{(\theta_0, \theta_1) \in \Theta} \sum_{t=1}^T (y_{0t} - \theta_0 -
    \theta_1'\by_t )^2.
\label{eq:fixeddidoracle}
\end{equation}
\eqref{eq:fixeddidoracle} simultaneously chooses the best intercept and the best set of
convex weights in hindsight. Because \eqref{eq:fixeddidoracle} is limited to using the
same intercept for prediction in each period, it is, in some sense, a
\emph{static} DID estimator. 

\Cref{thm:ftlregret} can be adapted to show that synthetic
control with an intercept is competitive against static DID.

\begin{prop}
\label{prop:staticdid}
    Consider demeaned synthetic control, where the analyst
    outputs the prediction $\hat y_t = \hat\theta_{0t} + \hat\theta_{t}'\by_t $ by
    solving the least-squares problem \[
\hat\theta_{0t}, \hat\theta_t = \argmin_{\theta_0, \theta \in [-2,2] \times \Delta^{N-1}} \sum_
{s < t} (y_{0s} - \theta_{0} -  \theta'\by_s)^2.
    \]
    Then, under bounded data $\norm{\bY}_\infty \le 1$, we have the following regret
    bound: \[
    \sum_{t=1}^T (y_{0t} - \hat y_{t})^2 - \min_{\theta_0, \theta \in [-2,2]\times
    \Delta^{N-1}} \sum_{t=1}^T  (y_{0s} - \theta_{0} - \theta'\by_s )^2 \le C N \log T
    \]
    for some constant $C$. 
\end{prop}

\begin{proof}
    We define the loss as $\frac{1}{2}(x-y)^2$, which only affects the regret up to a factor
of $2$.
\Cref{prop:staticdid} can be proved with \cref{thm:hazan}. Note that the diameter of
the
parameter space $[-2,2] \times \Delta^{N-1}$ can be bounded by $D = 2 \cdot \sqrt{2^2 +
 1} = 2\sqrt{5}$. The 2-norm of the vector $v_t = [1, \by_t']'$ is now  bounded by $R =
 \sqrt{N+1}$. The 1-norm of the parameter vector $\vartheta=[\theta_0,\theta']'$ is now
 bounded by
 $2+1=3$. Hence, $|v_t'\vartheta| \le 3$. Hence, we may take $b = 3 + 1 = 4$ and $a=1$.
 Plugging in, we obtain \[
 \reg_T \le 64N \bk{\log\pr{\frac{\sqrt 5}{2} \sqrt{N+1} T} + 1} < 
 CN\log T
 \]
 for some $C$. 
\end{proof}

\subsection{Proof of \cref{prop:twfe}}

Similarly to the proof of \cref{prop:staticdid}, suppose the adversary picks the
differences $|\tilde y_{it}| \le 2$, \emph{without} the constraint that the resulting
levels obey the restriction $
 \norm{\bY}_\infty \le 1$. An application of \cref{thm:hazan} shows that \[
 \sum_{t=1}^T (\tilde y_{0t} - \hat\theta_t' \tilde \by_t)^2 - \min_{\theta \in \Theta}
 \sum_{t=1}^T (\tilde
 y_{0t} - \theta' \tilde \by_t)^2 \le C N\log T
 \]
 for some $C$, uniformly over $|\tilde y_{it}| \le 2$, where $\hat \theta_t$ is the FTL
 strategy on the data $\tilde y_{it}$, which is exactly the synthetic control on the
 differenced data when $\bY$ is chosen by the adversary.

 Now, given any $\norm{\bY}_\infty \le 1$, we have that the corresponding differences
 $\tilde y_{it}$ obey the above regret bound, since they are bounded by $2$. Moreover, for
 both synthetic control
 ($\theta_t = \hat \theta_t$) and
 the oracle $\sigma_{\mathrm{TWFE}}$ ($\theta_t = \theta$), the prediction error of the
 data $y_{0t}$ is equal to the prediction error on the differences: \[
y_{0t} - \hat y_{t} = \frac{1}{t-1}\sum_{s< t} y_{0s} + \tilde y_{0t} - \pr{\frac{1}
{t-1}\sum_{s< t} y_{0s} + \theta_t'\tilde \by_t } = \tilde y_{0t} - \theta_t'\tilde
\by_t.
 \]
 Hence, we may rewrite the above regret bound as the bound \[
\sum_{t=1}^T (y_{0t} - \hat y_t)^2 - \min_{\theta \in \Theta} \sum_{t=1}^T (y_{0t} -
\hat y_t(\sigma_{\mathrm{TWFE}}(\theta)))^2 \le C N \log T.
    \]

\subsection{Proof of \cref{cor:regularized}}

\begin{theorem}
\label{thm:ftrl}
Assume that \begin{enumerate}[wide]
    \item $\ell_t(\theta) \equiv \ell(\theta'\by_t, y_{0t})$ is convex in $\theta$ for
    any $\bY$.
    \item The regularizer $\Phi(\theta)$ is $1$-strongly convex in some norm $
    \norm{\cdot}$. Normalize $\Phi$ such that its minimum over $\Theta$ is zero and
    maximum is $K < \infty$. 
    \item All subgradients $\nabla_\theta \ell_t(\theta)$ are bounded in the dual norm
    $\norm{\cdot}_*$, uniformly over $\Theta, \bY$: \[\norm{\nabla_\theta \ell_t
    (\theta)}^2_* \le G.\]
\end{enumerate}
Then FTRL attains the regret bound \[
\reg_T \le \frac{K}{\eta} + \frac{\eta T G }{2}.
\]
\end{theorem}

We first reproduce Corollary 7.9 from \cite{orabona2019modern} in our notation. Consider
an FTRL algorithm
that regularizes according to \[
\theta_t \in \argmin_{\theta} \sum_{s\le t} \ell_s(\theta) + \frac{1}{\eta} 
\Phi(\theta). 
\]
This corresponds to choosing $\eta_t = \eta$, $\psi(x) = \Phi(x)$, and $\min_\theta \Phi
(\theta) = 0$ in \cite{orabona2019modern}. 

\begin{theorem}[Corollary 7.9, \cite{orabona2019modern}]
    Let $\ell_t$ be a sequence of convex loss functions. Let $\Phi: \Theta \to \R$ be
    $\mu$-strongly convex with respect to the norm $\norm{\cdot}$. Then, FTRL guarantees 
    \[
    \sum_{t=1}^T \ell_t(\theta_t) - \sum_{t=1}^T \ell_t(\theta) \le \frac{\Phi(\theta)}
    {\eta} + \frac{\eta}{2\mu} \sum_{t=1}^T \norm{g_t}_*^2
    \]
    for all subgradients $g_t \in \partial \ell_t(\theta_t)$ and all $\theta \in
    \Theta$, where $\norm{\cdot}_*$ is the dual norm of $\norm{\cdot}$.
\end{theorem}

\begin{proof}[Proof of \cref{thm:ftrl}]
    \Cref{thm:ftrl} then follows immediately where $\norm{g_t}_*^2 \le G$, $\Phi(\theta)
\le K$, and $\mu = 1$. 
\end{proof}

\begin{proof}[Proof of \cref{cor:regularized}]

For both squared and absolute losses, we can bound the gradient of the loss function in
terms of \[
\norm{\nabla_\theta \ell_t(\theta)}_* = \norm{\nabla q(y-\hat y) \cdot \by_t}_* = |\nabla
q(y-\hat y)|\norm{\by_t}_* \le
2 \sup_{\norm{\by}_{\infty} \le 1} \norm{\by}_*
\]
under any norm, where $q(t) = t^2/2$ or $q(t) = |t|$.
This is because (i) for squared loss, the gradient $|\nabla f| = |y-\hat y|$ is bounded by
2 and
(ii) for
absolute loss, the subgradients $|\nabla f|$ are bounded by 1 and hence by 2. Hence, we
should pick $G$
to be $4 \sup_{\norm{\by}_\infty \le 1}\norm{\by}_*^2$.

For the quadratic penalty assumed, it is 1-strongly convex with respect to $
\norm{\cdot}_2$ by the assumption that the minimum eigenvalue of its Hessian is 1. Thus
the dual norm
$\norm{\cdot}_*$ is also the Euclidean norm, and we may take $G = 4N$. This yields the
bound by \cref{thm:ftrl}, since \[
\frac{K}{\sqrt{K(2TN)^{-1}}} + \frac{4N T}{2} \sqrt{\frac{K}{2TN}} = 2\sqrt{2} \sqrt{NTK}.
\]
Setting $K=1/2$ yields the ridge penalty result.

The entropy penalty is 1-strongly convex with respect to $\norm{\cdot}_1$.\footnote{This
is a well-known result in online convex optimization. To prove it, we first note that \[
\Phi(y) = \Phi(x) + \nabla \Phi(x)'(y-x) + D_{\mathrm{KL}}(y \Vert x),
\]
where $\Phi(x)=\sum_i x_i \log x_i + C$, $D_{\mathrm{KL}}(y \Vert x) = \sum_i y_i \log
(y_i/x_i)$, and $x,y$ lie in the interior of the simplex. Pinsker's inequality then
implies \[
\Phi(y) \ge \Phi(x) + \nabla \Phi(x)'(y-x) + \frac{1}{2} \norm{x-y}_1^2.
\]
This is exactly the definition of 1-strong convexity with respect to $\norm{\cdot}_1$.} Thus we may take $G = 4\norm{\by_t}_\infty^2 = 4$. The maximum of entropy
(shifted so that its minimum is zero) can take $K = \log N$. This yields the bound via \cref{thm:ftrl}.
\end{proof}

\subsection{Two-way fixed effect calculation}
\label{asub:twfe}
Consider the TWFE regression with known, nonnegative weights $\sum_{i=1}^N
w_i = 1$ and the normalization 
$w_0 = 1$: \[
\argmin_{\mu_i, \alpha_t} \sum_{\substack{i,t : (i,t)\neq (0, S)\\i\in \{0,\ldots,N\}
\\ t \in
[S]}} w_i (y_{it} - \mu_i - \alpha_t)^2.
\]
We may eliminate $(i,t) = (0,S)$ from the sum since $\lambda \one(i=0,
S=t)$ in \eqref{eq:twfe}
absorbs that term, leaving $\mu_i, \alpha_t$ unaffected.
Consider forecasting $y_{0S}$ with $\mu_0 + \alpha_S$ that solves the above
program. As a reminder, in this subsection, we show that the
estimated $\mu_0 + \alpha_S$ takes the form of forecasting with weighted
average on
differenced data.

The first-order condition for $\mu_i$ takes the form \[
\sum_{t=1}^{S-1} y_{it} - \mu_i - \alpha_t + \one(i\neq 0) (y_{iS} - \mu_i -
\alpha_t) = 0.
\]
Hence, \[
\mu_i = \begin{cases}
    \bar y_{i} - \bar\alpha & i \neq 0 \\ 
    \bar y_{0} - \frac{S}{S-1}\bar \alpha + \frac{1}{S-1} \alpha_S & i = 0
\end{cases}
\]
where $\bar \alpha = \frac{1}{S} \sum_{t=1}^S \alpha_t$ and $\bar y_{i}$ is the
sample mean of observations for unit $i$ over time $1,\ldots, S$, with the understanding
that $y_{0S}$ is not included for $\bar y_{0}$.
Hence, the forecast is $\mu_0 + \alpha_S = \bar y_{0} + \frac{S}{S-1} 
\pr{\alpha_S - \bar\alpha}.$

Let us inspect the first-order condition for $\alpha_S$: \[
\sum_{i=1}^N w_i (y_{iS} - \mu_i - \alpha_S) = \sum_{i=1}^N w_i (y_{iS} - \bar
y_
{i} + \bar\alpha - \alpha_S) = 0.
\]
Rearrange to obtain that $
\alpha_S - \bar \alpha = \sum_{i=1}^N w_i \pr{\frac{S-1}{S} y_{iS} - 
\frac{1}{S}\sum_{t=1}^{S-1} y_
{it}}.
$
Therefore, $
\frac{S}{S-1} (\alpha_S - \bar\alpha) = \sum_{i=1}^N w_i \pr{y_{iS} - 
\frac{1}{S-1} \sum_{t=1}^{S-1} y_{it}}.
$
Thus the forecast is \[
\mu_0 + \alpha_S = \frac{1}{S-1}\sum_{t=1}^{S-1} y_{0t} + \sum_{i=1}^N w_i \pr{y_{iS} - 
\frac{1}{S-1} \sum_{t=1}^{S-1} y_{it}}. 
\]
Note that arriving at this result does not use the fact that $w_0 =1$. Hence, $w_0$ does
not matter for $\mu_0 + \alpha_S$.

\section{Further extensions}
\label{asec:app2}
\subsection{Adaptive regret}
\label{sec:adaptivereg}

The online learning literature also has results for controlling the \emph{adaptive
regret}: \begin{equation}
    \mathrm{AdaptiveRegret}_T = \sup_{1\le r < s \le T} \sum_{t=r}^s \br{\ell_t(\theta_t) -
\min_{\theta_{r,s}} \sum_{t=r}^s \ell_t(\theta_{r,s})},
\end{equation}
which is the worst regret over any subinterval of $[T]$. An upper bound of adaptive
regret serves as an upper bound of the regret over any subperiod indexed by $r < s$. In
particular, suppose we obtain a $O(\log T)$ upper bound on adaptive regret, then we
obtain meaningful \emph{average} regret upper bounds for all subperiods significantly
longer than $O(\log T)$.

A simple meta-algorithm called \emph{Follow The Leading History} (FLH) \citep
[Algorithm 31 in][]{hazan2019introduction} serves as a wrapper for an online learning
algorithm $\sigma$, such that \begin{equation}
\mathrm{AdaptiveRegret}_T(\mathrm{FLH}(\sigma)) \le \reg_T(\sigma) + O(\log T).
\label{eq:adaptive}
\end{equation}
When applied to synthetic control, FLH takes the following form. We initialize $p_1^1 =
1$ and set $\alpha = \frac14$. At each time $t$, when
prompted to make a prediction about $y_{0t}$:
\begin{enumerate}
    \item Consider the synthetic control estimated weights $\theta_t^1,\ldots,
    \theta_t^{t},$ where $\theta_{t}^j$ is the synthetic control weights estimated based
    on data from \emph{time horizons} $j,\ldots, t-1$. 
    \item Output the weighted average $\theta_t = \sum_{j=1}^t p_t^j \theta_t^j$.
    \item After receiving $\by_t, y_{0t}$ (and hence receiving $\ell_t(\theta) = 
    \frac{1}{2}(y_
    {0t} -
    \theta'\by_t)^2$), instantiate \[
p_{t+1}^i \gets \frac{p_t^i e^{-\alpha \ell_t(\theta_t^i)}}{\sum_{j=1}^t p_t^j e^
{-\alpha
\ell_t(\theta_t^j)}} \quad 1 \le i \le t.
    \]
    \item Set $p_{t+1}^{t+1} = \frac{1}{t+1}$ and further update \[
    p_{t+1}^i \gets \pr{1-\frac{1}{t+1}} p_{t+1}^i \quad 1 \le i \le t.
    \]
\end{enumerate}
At each step, FLH applied to synthetic control continues to output a convex
weighted average of control unit outcomes, making it a type of synthetic control
algorithm. Theorem 10.5 in \cite{hazan2019introduction} then implies the bound
\eqref{eq:adaptive} for the above algorithm.\footnote{The proof follows immediately
 since $\frac{1}{2}(y_{0t} - \theta'\by_t)^2$ is $\frac{1}{4}$-exp-concave.
 That is, \[\theta \mapsto \exp\pr{-\frac{1}{4} \cdot \frac{1}{2} (y_{0t} -
 \theta'\by_t)^2}\] is concave. This is because $-2 \le
 y_{0t} -
\theta'\by_t
\le 2$, and $g(x) = \exp\pr{-\frac{1}{4} \cdot \frac{1}{2} x^2}$ is concave
on
$x\in [-2,2]$. The Hessian of $\exp\pr{-\frac{1}{4} \cdot \frac{1}{2} (y_
{0t} -
 \theta'\by_t)^2}$ in $\theta$ is then $g''(y_{0t} -
 \theta'\by_t) \by_t \by_t'$, which is negative semidefinite.} In a nutshell, FLH treats
 synthetic control
predictions
from different horizons
 as \emph{expert predictions}, and applies a no-regret online learning algorithm to
 aggregate these expert predictions. We direct readers to 
\cite{hazan2019introduction} for further intuitions about the algorithm.

Combined with \cref{thm:ftlregret} for synthetic control, we find that the adaptive
regret of FLH-synthetic control is of the same order $O(N \log T + N \log N)$. This
means that the average regret over any subperiod of length $T'$ is $O\pr{\frac{N\log T
+ N\log N}{T'}}$, a meaningful bound for long subperiods $T' \gg N\log T$. In other
  words, in a protocol where the adversary \emph{additionally} picks a subperiod of
  length $T'$, and nature subsequently samples a treatment timing uniformly randomly
  over the subperiod, FLH-synthetic control achieves expected regret bound of $O\pr
  {\frac{N\log T
+ N\log N}{T'}}.$ The adaptive regret bound thus partially relaxes the requirement for
  uniform treatment timing, and allows for expected regret control over random
  treatment timing on any subperiod.

\subsection{A note on inference}
\label{asec:inference}

\Copy{inference}{
 Under the treatment assignment model $S \sim \Unif[T]$, we may test the sharp null $H_0 :
 \by(1) = \by(0)$, leveraging symmetries arising from treatment assignment. This is
 similar in spirit to \cite{bottmer2021design}, who consider design-based inference under
 random assignment of the treated unit. They compute the variance of the estimated
 treatment effect (for treated unit $M \sim \Unif[N]$ at some fixed time $S$) under
 random assignment, holding the
 outcomes fixed, and propose an
 unbiased estimator. This is also similar in spirit to unit-randomization-based placebo
 tests \citep{abadie2010synthetic}.}

 Let $y_{t} = y_{0t}$ for $t < S$ and let $y_t = \by_t (1)$ for $t \ge S$ be the observed
 time series of the treated unit. For any prediction $\hat y_t$ that does not depend on
 $S$---not limited to synthetic control predictions---we may form the residuals $r_t =
 |y_t - \hat y_t|$. One (finite-sample) test of the sharp null rejects when $r_S$ is at
 least the $\lceil T (1-\alpha) \rceil$\th{} order statistic of the sample $\{r_1,\ldots,
 r_T\}.$ Since, under the null, $r_S$ is equally likely to equal any of $\{r_1,\ldots,
 r_T\}$, the probability of it being the among largest $100\alpha\%$ is bounded by
 $\alpha$.
 Similarly, if $S \sim \pi$ where $\pi_t \le C/T$, a least-favorable test may be
 constructed by rejecting when $r_t \ge r_{(T-\lfloor T \alpha / C\rfloor)}$. Informally
 speaking, this test is more powerful when the predictions $\hat y_t$ are better, and our
 regret guarantees are in this sense informative for inference. Moreover, note that
 this procedure is very similar to conformal inference 
 \citep{lei2018distribution,chernozhukov2021exact}. Conformal intervals rely
 on the assumption that the data is exchangeable in the underlying sampling
 process. This symmetry is true here by virtue of assuming
 $S \sim \Unif[T]$, since the treated period is equally likely to be any
 one.
 
 The argument above does not use the regret result. From Markov's inequality, we
 can control the probability for the prediction error to deviate far relative to its
 expectation \[
 \P_{S \sim \Unif[T]}\bk{(y_{0S} - \hat y_{S})^2 > c} \le 
 \frac{\E_S[\ell_S(\theta_S)]}{c} \le \frac{1}{c} \pr{\min_
 {\theta \in \Theta} \frac{1}{T}\sum_{i=1}^T \ell_t(\theta) + \frac{1} {T}\reg_T}.
 \] Under assumptions where the pre-treatment loss $\min_\theta \frac{1}{S-1} \sum_
  {t < S}
 \ell_t(\theta)$ is a consistent estimator for the oracle performance $\min_\theta
 \frac{1}{T}\sum_{i=1}^T \ell_t(\theta)$, the above observation allows for predictive
 confidence intervals for the untreated outcome and confidence intervals of the treatment
 effect, which are valid over random treatment timing.

\subsection{Risk interpretation under idiosyncratic errors}
\label{asec:means}
\Copy{means}{
We consider another interpretation of \eqref{eq:unconditionalrisk}. In
many data-generating processes, \[
\E_P\bk{\min_\theta \risk(\theta ,\bY, \by(1))}
\]
may not be small, because the realized data $\bY$ may contain certain
unforecastable components. The purpose of this section is to leverage the
decomposition \[
\E_P[(\hat y_{0t} - y_{0t})^2] = \E_P[\epsilon_t^2] + \E_P[(\hat y_{0t} -
\mu_t)^2],\]
where $\epsilon_t = y_{0t} - \mu_t$ is some unforecastable component
satisfying $\E_P[\epsilon_t \hat y_{0t}] = 0$.  This decomposition
breaks prediction errors into forecastable and unforecastable
components. Because of this additive decomposition, under certain
conditions on $\epsilon_t$, we can interpret risk differences as
regret on estimating the forecastable component $\mu_t$ (since $ \E_P
[\epsilon_t^2]$ cancels in the difference). We can also decompose risk into
the oracle error on estimating $\mu_t$, the regret against the oracle on
estimating $\mu_t$, and the variance of the unforecastable errors
$\epsilon_t$.

For a fixed $\theta$, under uniform
treatment timing we
have that \[
\E_P[\risk(\theta ,\bY, \by(1))] = \E_P[\E_S (y_{0S} - \mu_S)^2] +
\E_P[\E_S (\theta'\by_S - \mu_S)^2]
\]
for \emph{some} mean component $\mu_t$, possibly random, of the outcome process $y_
{0t}$. For instance, we may take $\mu_t = \E_P[y_{0t} \mid \bY_{1:t-1}, \by_t].$ For
this $\mu_t$, we can also write \[
\E_P[\risk(\sigma, \bY, \by(1))] =  \E_P[\E_S (y_{0S} - \mu_S)^2] +
\E_P[\E_S (\hat \theta_t'\by_S - \mu_S)^2],
\]
since $\hat \theta_t'\by_t$ depends solely on $\bY_{1:t-1}, \by_t$. We thus have the
following 
implication of \eqref{eq:unconditionalrisk} \[
\E_P[\E_S (\hat \theta_t'\by_S - \mu_S)^2] - \min_{\theta \in \Theta}\E_P
[\E_S (\theta'\by_S -
\mu_S)^2] \le \frac{1}{T} \sup_{\norm{\bY}_\infty \le 1} \reg_T(\sigma; \bY),
\]
which says that the risk difference of estimating the conditional mean
$\mu_t$---the forecastable component of the outcome process---is upper
bounded by the regret. As a corollary, if $P = P_T$ is a sequence of
data-generating processes where,
as $T\to\infty$,
\[\min_{\theta \in \Theta}\E_P
[\E_S (\theta'\by_S -
\mu_S)^2] \to 0,\]
then we obtain a consistency result for synthetic control, in that \[
\E_P[\E_S (\hat \theta_t'\by_S - \mu_S)^2] \to 0
\]
as well.

Shifting from risk differences to risks themselves, this means that the
treatment effect estimation risk for
synthetic control admits the following upper bound \begin{align*}
\E_P[\risk(\sigma, \bY, \by(1))] &\le \min_{\theta \in \Theta}\E_P
[\E_S (\theta'\by_S -
\mu_S)^2] +  \frac{1}{T} \sup_{\norm{\bY}_\infty \le 1} \reg_T(\sigma; \bY) \\ &\quad +
\E_P
[\E_S (y_{0S} - \mu_S)^2],
\end{align*}
where the first term is the best possible error on the forecastable
component $\mu_t$, the second term is the average regret, and the third
term is the variance of the unforecastable component that cannot be
improved upon. 
We think the first two terms are likely small, and the last term
is unavoidable.

This argument also extends to non-uniformly random treatment timing.
Suppose we have a joint distribution $Q$ of $(\bY, \by(1), S)$ such that
$\pi_t(\bY) = Q(S = t \mid \bY) \le C/T$. Suppose further that $y_{0t} = \mu_t +
\epsilon_t$, where $\E_Q[\epsilon_t \mid \mu_t, \pi_t, \bY_{1:t-1}, \by_t] = 0$ for some
mean
component $\mu_t$.\footnote{We can take $\mu_t = \E[y_{0t} \mid \by_t, \bY_{1:t-1}]$
whenever $S \indep \bY$ under $Q$.} Then we have a similar decomposition of the risk
of estimating the
treatment effect at $S$: \begin{align*}
\E_Q[(y_{0S} - \hat\theta_S'\by_t)^2] &= \sum_{t=1}^T \E_Q[\pi_t(\bY) (y_{0t}
- \hat\theta_S'\by_t)^2] \\ 
&= \sum_{t=1}^T \E_Q\bk{\pi_t(\bY)(y_{0t} - \mu_t)^2} + \E_Q[\pi_t(\bY)(\mu_t -
\hat\theta_t' \by_t)^2] \\ &\quad+ 2 \E_Q[\pi_t \epsilon_t (\mu_t -
\hat\theta_t' \by_t)] \\
&=   \E_Q[\epsilon_S^2] + \E_Q[(\mu_S - \hat\theta_S ' \by_S)^2] 
\tag{Last term is zero}\\
&\le \E_Q[\epsilon_S^2] + \frac{C}{T}\sum_{t=1}^T \E_Q[(\mu_t - \hat\theta_t'
\by_t)^2] \tag{$(1,\infty)$-H\"older's inequality}\\ 
& \le \E_Q[\epsilon_S^2] + C \Bigg(\min_{\theta \in \Theta} \frac{1}{T}\sum_{t=1}^T \E_Q
[(\mu_t -
\theta'
\by_t)^2] \\ &\hspace{8em} + \frac{1}{T} \sup_{\norm{\bY}_\infty \le 1} \reg_T(\sigma;
\bY)\Bigg).
\end{align*}
The last right-hand side is equal to the variance of the unforecastable component
$\epsilon_S$ plus $C$ times the oracle risk on estimating the mean component, as well
as $O(NT^{-1}\log T)$ regret.  If the oracle risk for estimating the mean component is
small, then synthetic control is
close to optimal, and its risk on estimating the mean component $\E_Q[(\mu_S -
\hat\theta_S ' \by_S)^2]$ is also small.\footnote{Note that the bound \[
\E_Q[(y_{0S} - \hat\theta_S'\by_t)^2] \le C\pr{\E_Q[\epsilon_S^2]  + \min_
{\theta
\in \Theta} \frac{1}{T}\sum_{t=1}^T \E_Q
[(\mu_t -
\theta'
\by_t)^2] + \frac{1}{T} \sup_{\norm{\bY}_\infty \le 1} \reg_T(\sigma; \bY)}
\]
is immediate and allows for $\mu_t = \E[y_{0t} \mid \bY_{1:t-1}, \by_t] = 0$, yet the
scaled idiosyncratic risk $C\E_Q[\epsilon_S^2] $ may be large.}}

\end{appendix}

\end{document}